\documentclass[11pt,a4paper]{amsart}

\usepackage[top=3cm,bottom=3cm,left=3cm,right=3cm,headsep=10pt,a4paper]{geometry}

\usepackage[utf8]{inputenc}
\usepackage[draft=false,colorlinks,citecolor=blue,linktocpage,breaklinks,hypertexnames=false,pdfpagelabels]{hyperref}
\usepackage{amsmath,amssymb,amsthm,dsfont}
\usepackage{physics}
\usepackage{cleveref}
\usepackage{graphicx}
\usepackage{xcolor}

\RequirePackage[framemethod=default]{mdframed}

\newmdenv[skipabove=7pt,
skipbelow=7pt,
backgroundcolor=darkblue!15,
innerleftmargin=5pt,
innerrightmargin=5pt,
innertopmargin=5pt,
leftmargin=0cm,
rightmargin=0cm,
innerbottommargin=5pt,
linewidth=1pt]{tBox}

\newmdenv[skipabove=7pt,
skipbelow=7pt,
backgroundcolor=darkkblue!15,
innerleftmargin=5pt,
innerrightmargin=5pt,
innertopmargin=5pt,
leftmargin=0cm,
rightmargin=0cm,
innerbottommargin=5pt,
linewidth=1pt]{sBox}

\newmdenv[skipabove=7pt,
skipbelow=7pt,
backgroundcolor=blue2!25,
innerleftmargin=5pt,
innerrightmargin=5pt,
innertopmargin=5pt,
leftmargin=0cm,
rightmargin=0cm,
innerbottommargin=5pt,
linewidth=1pt]{dBox}

\definecolor{darkblue}{RGB}{0,76,156}
\definecolor{darkkblue}{RGB}{0,0,153}
\definecolor{blue2}{RGB}{102,178,255}

\theoremstyle{plain}
\newtheorem{lemma}{Lemma}[]
\newtheorem{thm}[lemma]{Theorem}
\newtheorem{stp2}{Step}
\newtheorem{stp}{Step}
\newtheorem{prop}{Proposition}
\newtheorem{cor}{Corollary}
\theoremstyle{definition}
\newtheorem{definition}{Definition}
\theoremstyle{remark}
\newtheorem{remark}{Remark}

\newenvironment{theorem}{\begin{tBox}\begin{thm}}{\end{thm}\end{tBox}}

\newenvironment{step}{\begin{sBox}\begin{stp}}{\end{stp}\end{sBox}}

\newcommand{\identity}{\ensuremath{\mathds{1}}}
\newcommand{\hs}{\ensuremath{\mathcal H}}
\newcommand{\ent}[2]{D(#1||#2)}

\newcommand{\R}{\mathbb{R}}

\newcommand{\N}{\mathbb{N}}

\newcommand{\LLL}{\mathbb{L}}

\newcommand{\HH}{\mathcal{H}}
\newcommand{\BB}{\mathcal{B}}

\newcommand{\A}{\mathcal{A}}
\newcommand{\SSS}{\mathcal{S}}
\newcommand{\ds}{\displaystyle}
\def\ov{\overset}
\def\un{\underset}

\title{Superadditivity of quantum relative entropy for general states}

\author[Capel]{Ángela Capel}
\author[Lucia]{Angelo Lucia}
\author[Pérez-García]{David Pérez-García}

\address[Capel]{Instituto de Ciencias Matemáticas (CSIC-UAM-UC3M-UCM), C/ Nicolás Cabrera 13-15, Campus de Cantoblanco, 28049 Madrid, Spain}
\email{angela.capel@icmat.es}

\address[Lucia]{QMATH, Department of Mathematical Sciences, University of Copenhagen, Universitetsparken 5, 2100 Copenhagen, Denmark and NBIA, Niels Bohr Institute, University of Copenhagen, Blegdamsvej 17, 2100 Copenhagen, Denmark}

\email{angelo@math.ku.dk}

\address[Pérez-García]{Departamento de Análisis Matemático, Universidad Complutense de Madrid, 28040 Madrid, Spain and 
Instituto de Ciencias Matemáticas (CSIC-UAM-UC3M-UCM), C/ Nicolás Cabrera 13-15, Campus de Cantoblanco, 28049 Madrid, Spain}

\email{dperezga@ucm.es}

\date{\today}

\begin{document}
\maketitle

\begin{abstract}
The property of superadditivity of the quantum relative entropy states that, in a bipartite system $\hs_{AB}=\hs_A \otimes \hs_B$, for every density operator $\rho_{AB}$ one has $\ent{\rho_{AB}}{\sigma_A\otimes \sigma_B} \ge \ent{\rho_A}{\sigma_A}+\ent{\rho_B}{\sigma_B}$. In this work, we provide an extension of this inequality for arbitrary density operators $\sigma_{AB}$. More specifically, we prove that  $ \alpha (\sigma_{AB})\cdot \ent{\rho_{AB}}{\sigma_{AB}} \ge \ent{\rho_A}{\sigma_A}+\ent{\rho_B}{\sigma_B}$ holds for all bipartite states $\rho_{AB}$ and $\sigma_{AB}$, where $\alpha(\sigma_{AB})= 1+2 \norm{\sigma_A^{-1/2} \otimes \sigma_B^{-1/2} \, \sigma_{AB} \, \sigma_A^{-1/2} \otimes \sigma_B^{-1/2} - \identity_{AB}}_\infty$.
\end{abstract}

\section{Introduction and notation}\label{sec-1}

The \textit{quantum relative entropy} between two density operators $\rho$ and $\sigma$ in a finite dimensional Hilbert space, $\ent \rho \sigma$, is given by $\tr[\rho (\log \rho - \log \sigma)]$ if $\text{supp}(\rho) \subseteq \text{supp}(\sigma) $ and by $+ \infty$ otherwise\footnote{It can also be defined in infinite dimensions, as well as generalized von Neumann algebras \cite{libropetz}. However, in this work, for simplicity we will restrict to finite dimensions.}. It constitutes a measure of distinguishability between two quantum states and is a fundamental tool in quantum information theory \cite{libropetz}, \cite{wilde}.

The quantum relative entropy is the quantum analogue of the Kullback-Leibler divergence \cite{kld}, the probabilistic relative entropy. Its origin lies in mathematical statistics, where it is used to measure how much two states differ in the sense of statistical distinguishability. The larger the relative entropy of two states is, the more information for discriminating between the hypotheses associated to them can be obtained from an observation.

One of the main properties of quantum relative entropy is \textit{superadditivity}, which states that in a bipartite system $\hs_{AB}=\hs_A \otimes \hs_B$ one has:
\begin{equation}\label{superadditivity}
\ent{\rho_{AB}}{\sigma_A\otimes \sigma_B} \ge \ent{\rho_A}{\sigma_A}+\ent{\rho_B}{\sigma_B}
\end{equation}
for all $\rho_{AB}$, where we use the standard notation $\rho_A=\tr_B[\rho_{AB}]$ and $\tr_B$ is the partial trace.

Since  (Proposition \ref{prop:sigmaprod})
\begin{equation*}
\ent{\rho_{AB}}{\sigma_A\otimes \sigma_B} -\ent{\rho_A}{\sigma_A}-\ent{\rho_B}{\sigma_B}= \ent{\rho_{AB}}{\rho_A\otimes \rho_B},
\end{equation*}
\eqref{superadditivity} is equivalent to the fact that the mutual information $I_\rho(A:B):= \ent{\rho_{AB}}{\rho_A\otimes \rho_B}$ is always non-negative, a fact that appears ubiquitously in quantum information theory. 

In the form \eqref{superadditivity}, superadditivity of the quantum relative entropy has found applications in e.g. quantum thermodynamics \cite{gallego}, statistical physics \cite[Chapter 13]{libropetz} or hypothesis testing \cite{petz}. Indeed, as proven recently in \cite{axcharRE} (building on results from \cite{matsumoto}), the property of superadditivity, along with the properties of \textit{continuity} with respect to the first variable, \textit{monotonicity} and \textit{additivity} (Proposition \ref{prop:REprop}), characterizes axiomatically the quantum relative entropy. 

The main aim of this work is to provide a quantitative extension of (\ref{superadditivity}) for an arbitrary density operator $\sigma_{AB}$. Note that for all $\rho_{AB}$ and $\sigma_{AB}$, as a consequence of monotonicity of the quantum relative entropy for the partial trace, the following holds:
\begin{equation}\label{monotonicity}
2\ent{\rho_{AB}}{\sigma_{AB}}\ge \ent{\rho_A}{\sigma_A} + \ent{\rho_B}{\sigma_B}.
\end{equation}

Therefore we aim to give a constant $\alpha (\sigma_{AB}) \in [1,2]$ at the LHS of (\ref{superadditivity}) that measures how far $\sigma_{AB}$ is from $\sigma_A \otimes \sigma_B$.

Following \cite{clasico} we will consider as $\alpha(\sigma_{AB})-1$ the distance from $\identity$ to ``$\sigma_{AB}$  {\it multiplied by the inverse of $\sigma_A\otimes\sigma_B$}''.
In the case in which $\sigma_{AB}$ and $\sigma_A\otimes\sigma_B$ commute there is a unique way to define this: $\sigma_{AB} \,(\sigma_A^{-1}\otimes \sigma_B^{-1})$. In the non-commutative case, however, there are many possible ways to define the multiplication by the inverse.  The one we will take in the result below is a symmetric analogue of the commutative case, $(\sigma_A^{-1/2}\otimes \sigma_B^{-1/2}) \,  \sigma_{AB} \, (\sigma_A^{-1/2}\otimes \sigma_B^{-1/2})$.
Another one that will appear in the proof of this result is the derivative of the matrix logarithm on $\sigma_A\otimes \sigma_B$ evaluated on $\sigma_{AB}$, $\mathcal{T}_{\sigma_A\otimes \sigma_B}(\sigma_{AB})$, whose explicit equivalent expressions shown in \cite{lieb} and \cite{sutter} will be presented later.

\begin{theorem}
\label{thm:quasifactorizationAB} For any bipartite states $\rho_{AB},\sigma_{AB}$:
\begin{equation*}\label{eq:superadditivity}
(1+2\|H(\sigma_{AB})\|_{\infty})\ent{\rho_{AB}}{\sigma_{AB}}\ge \ent{\rho_A}{\sigma_A} + \ent{\rho_B}{\sigma_B},
\end{equation*}
where

\begin{center}
$H(\sigma_{AB}) =  \sigma_A^{-1/2} \otimes \sigma_B^{-1/2} \, \sigma_{AB} \, \sigma_A^{-1/2} \otimes \sigma_B^{-1/2} - \identity_{AB}$,
\end{center}
and $\identity_{AB} $ denotes the identity operator in $\hs_{AB}$.

Note that $H(\sigma_{AB})=0$ if $\sigma_{AB}=\sigma_A\otimes \sigma_B$.
\end{theorem}

This result constitutes an improvement to (\ref{monotonicity}) whenever $ \|H (\sigma_{AB}) \|_{\infty} \leq 1/2 $ (and, hence, $ 1+2\|H (\sigma_{AB}) \|_{\infty} \leq 2 $). Then, it is likely to be relevant for situations where it is natural to assume $\sigma_{AB} \sim \sigma_A \otimes \sigma_B$. This is the case of (quantum) many body systems where such property is expected to hold for spatially separated regions $A, B$ in the Gibbs state above the critical temperature. Indeed, a classical version of Theorem \ref{thm:quasifactorizationAB} proven by Cesi \cite{cesi} and Dai Pra, Paganoni and Posta \cite{clasico}, was the key step to provide the arguably simplest proof of the seminal result of Martinelli and Olivieri \cite{marti-oliv} connecting the decay of correlations in the Gibbs state of a classical spin model with the mixing time of the associated Glauber dynamics, via a bound on the log-Sobolev constant.

\subsection{Notation}

We consider a finite dimensional Hilbert space $\HH$. We denote the set of bounded linear operators acting on $\HH$ by $\BB= \BB(\hs)$ (whose elements we denote by lowercase Latin letters: $f,g$...), and its subset of Hermitian operators by $\A \subseteq \BB$ (whose elements we call \textit{observables}). The set of positive semidefinite Hermitian operators is denoted by $\A^+$. We also denote the set of density operators by $\SSS= \qty{f \in \A^+ \, : \, \tr[f]=1}  $ (whose elements we also call \textit{states} and denote by lowercase Greek letters: $\sigma, \rho$...).

 A linear map $\mathcal{T}: \BB \rightarrow \BB $ is called a \textit{superoperator}. We say that a superoperator $\mathcal{T}$ is \textit{positive} if it maps positive operators to positive operators. Moreover, we denote $\mathcal{T}$ as \textit{completely positive} if $\mathcal{T} \otimes \identity : \mathcal{B} \otimes \mathcal{M}_n \rightarrow \mathcal{B} \otimes \mathcal{M}_n $ is positive for every $n \in \N$, where $ \mathcal{M}_n $ is the space of complex $n \times n$ matrices. We also say that $\mathcal{T}$ is \textit{trace preserving} if $\tr[\mathcal{T}(f)]= \tr[f]$ for all $f \in \BB$. Finally, if $\mathcal{T}$ verifies all these properties, i.e., is a completely positive and trace preserving map, it is called a \textit{quantum channel} (for more information on this topic, see \cite{wolf}).
 
 We denote by $\norm{\cdot}_1$ the trace norm  $\left( \norm{f}_1 = \tr[\sqrt{f^* f}] \right)$  and by $\norm{\cdot}_\infty$ the operator norm $\ds \left( \norm{f}_\infty = \text{sup} \qty{ \norm{f(x)}_{\hs}  :  \norm{x}_{\hs}=1 }\right)$. In the following section, we will make use of this (Hölder) inequality \cite{bhatia}:
\begin{equation}
\norm{f g}_1 \leq \norm{f}_1 \norm{g}_\infty \text{ \phantom{asfs} for every }f, g \in \BB.
\end{equation}

In most of the paper, we consider a bipartite finite dimensional Hilbert space $\HH_{AB}= \HH_A \otimes \HH_B$. When this is the case, we use the previous notation placing the subindex ${AB}$ (resp. $A$, $B$) in each of the previous sets to denote that the operators considered act on $\hs_{AB}$ (resp. $\hs_A$, $\hs_B$). There is a natural inclusion of $\A_{A}$ in $\A_{AB}$ by identifying $\A_{A} = \A_{A} \otimes \identity_{B} $.   


\subsection{Relative entropy}

Let $\hs$ be a finite dimensional Hilbert space, $f,g\in \A^+$, $f$ verifying $\tr[f] \neq 0$.  The \textit{quantum relative entropy} of $f$ and $g$ is defined by \cite{umegaki}:
\begin{equation}
\ent f g = \frac{1}{\tr[ f]}\tr \left[ f (\log f - \log g) \right].
\end{equation}

\begin{remark}
In most of the paper we only consider density matrices (with trace $1$). Let $\rho	, \sigma \in \SSS$. In this case, the quantum relative entropy is given by:
\begin{equation}
\ent {\rho} {\sigma} = \tr \left[ {\rho} (\log \rho - \log \sigma) \right].
\end{equation}
\end{remark}

In the following proposition, we collect some well-known properties of the relative entropy, which will be of use in the following section.

\begin{prop}[Properties of the relative entropy, \cite{wehrl}]\label{prop:REprop} \text{\phantom{s}}\\
Let $\hs_{AB}$ be a bipartite finite dimensional Hilbert space, $\hs_{AB} = \hs_A \otimes \hs_B$. Let $\rho_{AB}, \sigma_{AB} \in \SSS_{AB}$. The following properties hold:
\begin{enumerate}
\item \textbf{Non-negativity.} $\ent {\rho_{AB}} {\sigma_{AB}} \geq 0$ and $\ent {\rho_{AB}} {\sigma_{AB}}=0 \Leftrightarrow {\rho_{AB}}=\sigma_{AB}$.
\item \textbf{Finiteness.} $\ent {\rho_{AB}} {\sigma_{AB}} < \infty $ if, and only if, $\text{supp}(\rho_{AB}) \subseteq \text{supp}(\sigma_{AB})$, where supp stands for support.
\item \textbf{Monotonicity.} $\ent {\rho_{AB}} {\sigma_{AB}} \geq \ent {T(\rho_{AB})} {T(\sigma_{AB})}$ for every quantum channel $T$.
\item \textbf{Additivity.} $\ent {\rho_A \otimes \rho_B} {\sigma_A \otimes \sigma_B}= \ent {\rho_A} {\sigma_A} + \ent {\rho_B} {\sigma_B}$.
\end{enumerate}
\end{prop}

These properties, especially the property of non-negativity, allow to consider the relative entropy as a measure of separation of two states, even though, technically, it is not a distance (with its usual meaning), since it is not symmetric and lacks a triangle inequality.

Let us prove now the property of superadditivity, whenever $\sigma_{AB}= \sigma_A \otimes \sigma_B$.

\begin{prop}\label{prop:sigmaprod}

Let $\hs_{AB}=\hs_A \otimes \hs_B$ and  $\rho_{AB}, \sigma_{AB} \in \SSS_{AB}$. If $\sigma_{AB}= \sigma_A \otimes \sigma_B$, then
\begin{center}
 $\ent {\rho_{AB}} {\sigma_{AB}} = I_\rho(A:B) + \ent {\rho_A} {\sigma_A}+\ent {\rho_B} {\sigma_B}$,
\end{center}
where $I_\rho(A:B)=\ent {\rho_{AB}} {\rho_A \otimes \rho_B}$ is the mutual information \cite{shannon}. 

As a consequence,
\begin{center}
$\ds \ent {\rho_{AB}} {\sigma_{A}\otimes \sigma_B } \geq \ent {\rho_A} {\sigma_A}+\ent {\rho_B} {\sigma_B}$.
\end{center}

\end{prop}

\begin{proof}
Since $\sigma_{AB}= \sigma_A \otimes \sigma_B$, we have
\begin{eqnarray*}
\ent {\rho_{AB}} {\sigma_{A} \otimes \sigma_{B}} &=& \tr[\rho_{AB}(\log \rho_{AB} - \log \sigma_{A} \otimes \sigma_{B})]\\
&=& \tr[\rho_{AB} (\log \rho_{AB} - \log \rho_{A} \otimes \rho_{B} +\log \rho_{A} \otimes \rho_{B}  - \log \sigma_{A} \otimes \sigma_{B})] \\
&=& \ent {\rho_{AB}} {\rho_A \otimes \rho_B} + \ent {\rho_A \otimes \rho_B} { \sigma_{A} \otimes \sigma_{B}}\\
&=& I_\rho (A:B) + \ent {\rho_A} {\sigma_A} + \ent{\rho_B} {\sigma_B}.
\end{eqnarray*}

Now, since $I_\rho (A:B)$ is a relative entropy, it is greater or equal than zero (property 1 of Proposition \ref{prop:REprop}), so 
\begin{center}
$\ds \ent {\rho_{AB}} {\sigma_{A}\otimes \sigma_B } \geq \ent {\rho_A} {\sigma_A}+\ent {\rho_B} {\sigma_B}$.
\end{center}
\end{proof}
We prove now a lemma for observables (non necessarily of trace $1$) which yields a relation between the relative entropy of two observables and the relative entropy of some dilations of each of them. In particular, it is a useful tool to express the relative entropy of two observables in terms of the relative entropy of their normalizations (i.e., the quotient of each of them by their trace).

\begin{lemma}\label{lemma:rel-entropy-homogeneity}
Let $\hs$ be a finite dimensional Hilbert space and let $f,g\in \A^+$ such that $\tr[f] \neq 0$. For all positive real numbers $a$ and $b$, we have:
\begin{equation}
		\ent {af} {bg} = \ent f g + \log \frac{a}{b}.
\end{equation}
\begin{proof}
\begin{eqnarray*}
\ent{af}{bg} &=& \frac{1}{a \tr f} \left(a \tr \left[ f \left( \log a f - \log bg \right) \right] \right) \\
&=& \frac{1}{\tr f} \qty(\tr[ f \log a] + \tr[ f \log f] - \tr[ f \log b] - \tr[f \log g])  \\
&=&\frac{1}{\tr f} (\tr[ f \left( \log f -  \log g \right) ] ) + \log a - \log b \\
&=& \ent f g + \log \frac{a}{b},
\end{eqnarray*}
where, in the first and third equality, we are using the linearity of the trace, and we are denoting $\log a \identity$ by $\log a$ for every $a\geq 0$.
\end{proof}
\end{lemma}

Since the relative entropy of two density matrices is non-negative (property 1 of Proposition \ref{prop:REprop}), we have the following corollary:

\begin{cor}\label{lemma:cond-ent}
Let $\hs$ be a finite dimensional Hilbert space and let $f,g\in \A^+$ such that $\tr[f] \neq 0$ and $\tr[g] \neq 0$. Then, the following inequality holds:
\begin{equation}\label{eq:cond-ent}
\ent f g \ge - \log \frac{\tr[ g]}{\tr[ f]}.
\end{equation}
\end{cor}
\begin{proof}
Since $f/\tr[f]$ and $g / \tr[g]$ are density matrices, we have that
\begin{center}
$\ent {f/\tr[f] \,} {\, g / \tr[g]} \ge 0 $,
\end{center}
and we can apply Lemma \ref{lemma:rel-entropy-homogeneity}:
\begin{center}
$\ds 0 \leq \ent {f/\tr[f] \,} {\, g / \tr[g]} = \ent f g + \log \frac{\tr[g]}{\tr[f]} $.
\end{center}
\vspace{-0.2cm}
\end{proof}

\section{Proof of main result}\label{sec:results}

We divide the proof of Theorem \ref{thm:quasifactorizationAB} in four steps. 

In the first step, we provide a lower bound for the relative entropy of $\rho_{AB}$ and $\sigma_{AB}$ in terms of $\ent {\rho_{A}} {\sigma_{A}}$,  $\ent {\rho_{B}} {\sigma_{B}} $ and an error term, which we will further bound in the following steps.

\begin{step}
\label{step:1}
For density matrices  $\rho_{AB}, \sigma_{AB} \in \SSS_{AB}$, it holds that
\begin{equation}\label{eq:step-1}
\ent{\rho_{AB}}{\sigma_{AB}} \ge \ent {\rho_{A}} {\sigma_{A}} + \ent {\rho_{B}} {\sigma_{B}} - \log \tr M,
\end{equation}
where
$  M = \exp \bqty{ \log \sigma_{AB}  - \log \sigma_{A} \otimes \sigma_{B} +\log \rho_{A} \otimes \rho_{B} } $
and equality holds (both sides being equal to zero) if $\rho_{AB} =
\sigma_{AB}$. \\
Moreover, if $\sigma_{AB} = \sigma_A \otimes \sigma_B$, then $\log \tr M =0$.
\end{step}

\begin{proof}
It holds that:
\begin{eqnarray*}
\ds \ent {\rho_{AB}} {\sigma_{AB}} & - & \left[ \ent {\rho_{A}} {\sigma_{A}} + \ent {\rho_{B}} {\sigma_{B}}  \right]=\\
&=& \ent {\rho_{AB}} {\sigma_{AB}} - \ent {\rho_{A} \otimes \rho_{B}} {\sigma_{A} \otimes \sigma_{B}} \\
&=& \tr \left[ {\rho_{AB}} \left(  \log {\rho_{AB}} - \underbrace{ \left( \, \log {\sigma_{AB}}  - \log \sigma_{A} \otimes \sigma_{B} +\log \rho_{A} \otimes \rho_{B} \, \right) }_{\log M} \right) \right] \\
& =& \ent {\rho_{AB}} M,
\end{eqnarray*}
where $M$ is defined as in the statement of the step and in the first equality we have used the fourth property of Proposition \ref{prop:REprop}.

\vspace{0.2cm}

We can now apply Corollary \ref{lemma:cond-ent} to obtain that
\begin{center}
$\ent {\rho_{AB}} M = \tr[ {\rho_{AB}} (\log {\rho_{AB}} - \log M) ] \ge - \log \tr M$.
\end{center}

It is easy to check, given the definition of $M$, that $M=\sigma_{AB}$ if $\rho_{AB} = \sigma_{AB}$, so both sides are equal to zero in this case.

Also, if $\sigma_{AB}= \sigma_A \otimes \sigma_B$, $M$ is equal to $\rho_A\otimes \rho_B$.
In both cases we have $\log \tr M = 0$.
\end{proof}

Our target now is to bound the error term, $\log \tr M$, in terms of the relative entropy of $\rho_{AB}$ and $\sigma_{AB}$ times a constant which depends only on $A$, $B$ and $\sigma_{AB}$, and represents how far $\sigma_{AB}$ is from being a tensor product. In the second step of the proof, we will bound this term by the trace of the product of a term which contains this `distance' between $\sigma_{AB}$ and $\sigma_{A} \otimes \sigma_{B}$ and another term which depends on $\rho_{AB}$ and not on $\sigma_{AB}$. However, before that, we need to introduce some concepts and results.

First, we recall the Golden-Thompson inequality, proven independently in \cite{golden} and \cite{thompson} (and extended to the infinite dimensional case in \cite{GTinfdim} and \cite{GTinfdim2}), which says that for Hermitian operators $f$ and $g$,
\begin{equation}
\tr[e^{f+g}] \leq \tr[e^f e^g],
\end{equation} 
where we denote by $e^f$ the exponential of $f$, given by
\begin{center}
$\ds e^f:= \ov{\infty}{\un{k=0}{\sum}} \frac{f^k}{k!}  $. 
\end{center}

The trivial generalization of the Golden-Thompson inequality to three operators instead of two in the form $\tr[e^{f+g+h}] \leq \tr[e^f e^g e^h]$ is false, as Lieb mentioned in \cite{lieb}. However, in the same paper, he provides a correct generalization of this inequality for three operators. This result has recently been extended by Sutter et al. in \cite{sutter} via de so-called multivariate trace inequalities (see also the subsequent paper by Wilde \cite{multioperator2}, where similar inequalities are derived following the statements of \cite{multioperator1}).

\begin{thm}
\label{thm:Lieb}
Let $f, g$ be positive semidefinite operators, and recall the definition of $\mathcal{T}_g$:
\begin{equation}
	\mathcal T_g(f) = \int_0^\infty \dd{t} (g+t)^{-1} f (g+t)^{-1} .
\end{equation}
$\mathcal T_g$ is positive semidefinite if $g$ is. 
We have that
\begin{equation}
	\tr[ \exp(-f+g+h)] \le \tr[ e^h \mathcal T_{e^f}(e^g)].
\end{equation}
\end{thm}

This superoperator $\mathcal T_g$ provides a pseudo-inversion of the operator $g$ with respect to the operator where it is evaluated. In particular, if $f$ and $g$ commute, it is exactly the standard inversion, as we can see in the following corollary. 

\begin{cor}
If $f$ and $g$ commute, then
\[ \mathcal T_g(f) = f \int_0^\infty \dd{t} (g + t)^{-2} = f g^{-1}, \]
and therefore
\[ \tr[ \exp(-f+g+h) ]\le \tr[ e^h e^{-f} e^g ] = \tr[ e^h e^{-f+g}] .\]
This shows that Lieb's theorem is really a generalization of Golden-Thompson inequality.
\end{cor}

We use an alternative definition of this superoperator to obtain a necessary tool for the proof of Step \ref{step:2}. In \cite[Lemma 3.4]{sutter}, Sutter, Berta and Tomamichel prove the following result:

\begin{lemma}
For  ${f} $ a positive semidefinite operator and $g $ a Hermitian operator the following holds:
\begin{center}
$\ds \mathcal{T}_g(f)= \int_{-\infty}^\infty dt \, \beta_0 (t) \, g^{\frac{-1-it}{2}} \,f \, g^{\frac{-1+it}{2}} ,$
\end{center}
with
 \begin{center}
 $\ds \beta_0(t)= \frac{\pi}{2} (\cosh(\pi t)+ 1)^{-1}  $.
 \end{center}
\end{lemma} 

Using this expression for $\mathcal{T}_{\sigma_A \otimes \sigma_B} (\sigma_{AB})$, we can prove the following result, which is a quantum version of a result used in \cite{clasico}.

\begin{lemma}\label{lemma:remark}
For every operator $O_{A} \in \BB_A$ and $O_{B} \in \BB_B$ the following holds:
\begin{center}
$\ds \tr[ L(\sigma_{AB}) \, \sigma_{A} \otimes O_{B} ]= \tr[ L(\sigma_{AB}) \, O_{A} \otimes
  \sigma_{B} ]= 0,$
\end{center}
where
\begin{center}
$\ds   L(\sigma_{AB}) = \mathcal{T}_{\sigma_{A} \otimes \sigma_{B}} \left(\sigma_{AB} \right)  - \identity_{AB}$.
\end{center} 

\end{lemma}
\begin{proof}
We only prove 
\begin{center}
$\tr[ L(\sigma_{AB}) \, \sigma_{A} \otimes O_{B} ]=0$,
\end{center}
since the other equality is completely analogous.
\begin{align*}
& \tr[ L(\sigma_{AB}) \, \sigma_{A} \otimes O_{B} ]=\\
&\phantom{asdad} = \tr[ \left( \mathcal{T}_{\sigma_{A} \otimes \sigma_{B}} \left(\sigma_{AB} \right)  - \identity_{AB}\right) \sigma_{A} \otimes O_{B}]\\
&\phantom{asdad}=\tr[ \mathcal{T}_{\sigma_{A} \otimes \sigma_{B}} \left(\sigma_{AB} \right)  \sigma_{A} \otimes O_{B}]- \tr[\sigma_{A} \otimes O_{B}] \\
&\phantom{asdad}= \tr[\int_{- \infty}^\infty dt \, \beta_0 (t) \left( \sigma_A \otimes \sigma_B \right)^{\frac{-1-it}{2}} \, \sigma_{AB} \,  \left( \sigma_A \otimes \sigma_B \right)^{\frac{-1+it}{2}} \,  \sigma_{A} \otimes O_{B}] - \tr[O_B]\\
&\phantom{asdad}= \int_{- \infty}^\infty dt \, \beta_0 (t) \tr[ \sigma_A^{\frac{-1-it}{2}} \otimes \sigma_B^{\frac{-1-it}{2}} \, \sigma_{AB} \,  \sigma_A^{\frac{-1+it}{2}} \otimes \sigma_B^{\frac{-1+it}{2}}  \,  \sigma_{A} \otimes O_{B}]-\tr[O_{B}],
\end{align*}
because $\tr[\sigma_A]=1$, the integral commutes with the trace, $\beta_0(t)$ is a scalar for every $t \in \R$ and the exponent in the power of a tensor product can be split into both terms.

Now, since the trace is cyclic and using the fact that any operator in $\hs_B$ commutes with every operator in $\hs_A$, we have:
\begin{align*}
& \tr[ L(\sigma_{AB}) \, \sigma_{A} \otimes O_{B} ]=\\
&\phantom{asdad}= \int_{- \infty}^\infty dt \, \beta_0 (t) \tr[  \sigma_{AB} \,  \sigma_A^{\frac{-1+it}{2}} \otimes \sigma_B^{\frac{-1+it}{2}}  \,  \sigma_{A} \otimes O_{B} \, \sigma_A^{\frac{-1-it}{2}} \otimes \sigma_B^{\frac{-1-it}{2}}]-\tr[O_{B}]\\
&\phantom{asdad}= \int_{- \infty}^\infty dt \, \beta_0 (t) \tr[  \sigma_{AB} \, \left( \sigma_A^{\frac{-1+it}{2}} \sigma_{A} \, \sigma_A^{\frac{-1-it}{2}} \right) \otimes \left( \sigma_B^{\frac{-1+it}{2}}  O_{B}  \,\sigma_B^{\frac{-1-it}{2}} \right)]-\tr[O_{B}]\\
&\phantom{asdad}= \int_{- \infty}^\infty dt \, \beta_0 (t) \tr[  \sigma_{AB} \, \identity_A \otimes  \left( \sigma_B^{\frac{-1+it}{2}}  O_{B}  \,\sigma_B^{\frac{-1-it}{2}} \right)]-\tr[O_{B}]\\
&\phantom{asdad}= \int_{- \infty}^\infty dt \, \beta_0 (t) \tr[  \sigma_{B} \,  \sigma_B^{\frac{-1+it}{2}}  O_{B}  \,\sigma_B^{\frac{-1-it}{2}} ]-\tr[O_{B}]\\
&\phantom{asdad}= \int_{- \infty}^\infty dt \, \beta_0 (t) \tr[ \sigma_B^{\frac{-1-it}{2}} \, \sigma_{B} \,  \sigma_B^{\frac{-1+it}{2}}  O_{B}  ]-\tr[O_{B}]\\
&\phantom{asdad} =  \tr[ O_{B}] \int_{- \infty}^\infty dt \, \beta_0 (t) -\tr[O_{B}]\\
&\phantom{asdad}= 0,
\end{align*}
where we have used 
\begin{center}
$\ds \int_{- \infty}^\infty dt \, \beta_0 (t)=1  $,
\end{center}
and the fact that, for every $f_A \in \BB_A$ and $g_{AB} \in \SSS_{AB}$, the following holds: 
\begin{center}
$\ds \tr[f_A \otimes \identity_B \, g_{AB}] = \tr[f_A \, g_A]  $.
\end{center}
\end{proof}

We are now in position to develop the second step of the proof.

\begin{step}\label{step:2}
With the same notation of \cref{step:1}, we have that
\begin{equation}
\log \tr M \le  \tr[L(\sigma_{AB})  \left( \rho_{A} - \sigma_A \right) \otimes \left(  \rho_{B} - \sigma_B \right) ],
\end{equation}
where
\begin{center}
$\ds   L(\sigma_{AB})= \mathcal{T}_{\sigma_{A} \otimes \sigma_{B}} \left(\sigma_{AB} \right)  - \identity_{AB}$.
\end{center} 
\end{step}

\begin{proof}
We apply Lieb's theorem to the error term of inequality (\ref{eq:step-1}):
\begin{eqnarray*}
\tr M &=&
\tr\left[ \exp(  \underbrace{\log \sigma_{AB} }_{g} -
\underbrace{\log \sigma_{A} \otimes \sigma_{B} }_{f} +
\underbrace{\log \rho_{A} \otimes \rho_{B} }_{h} ) \right] \\
&\leq & \tr[ \rho_{A} \otimes \rho_{B} \mathcal T_{\sigma_{A} \otimes \sigma_{B}} ( \sigma_{AB})]\\
&=&  \tr[  \rho_{A} \otimes \rho_{B} \underbrace{\left( \mathcal T_{\sigma_{A} \otimes \sigma_{B}} (\sigma_{AB}) -\identity_{AB} \right)}_{L(\sigma_{AB})}] + \underbrace{\tr[ \rho_{A} \otimes \rho_{B}]}_1,
\end{eqnarray*}
where we are adding and substracting $\rho_A \otimes \rho_B$ inside the trace in the last equality.

Now, using the fact $\log(x)\le x-1$, we have
\begin{center}
$\log \tr M \leq \tr M - 1 \leq  \tr[ L(\sigma_{AB}) \, \rho_{A} \otimes \rho_{B} ]$.
\end{center}

Finally, in virtue of Lemma \ref{lemma:remark}, it is clear that
\begin{center}
$ \ds \tr[ L(\sigma_{AB}) \, \rho_{A} \otimes \rho_{B} ]= \tr[L(\sigma_{AB})  \left( \rho_{A} - \sigma_A \right) \otimes \left(  \rho_{B} - \sigma_B \right) ] $.
\end{center}

Therefore,
\begin{center}
$ \log \tr M \leq  \tr[L(\sigma_{AB})  \left( \rho_{A} - \sigma_A \right) \otimes \left(  \rho_{B} - \sigma_B \right) ]$.
\end{center}

Notice that if $\sigma_{AB}= \sigma_A \otimes \sigma_B$, then $\mathcal T_{\sigma_{A} \otimes \sigma_{B}} ( \sigma_{AB})= \left( \sigma_{A} \otimes \sigma_{B} \right)^{-1} \sigma_{A} \otimes \sigma_{B}= \identity_{AB} $, so $L(\sigma_{AB})=0$.
\end{proof}

In the third step of the proof, we need to bound $\tr[L(\sigma_{AB})  \left( \rho_{A} - \sigma_A \right) \otimes \left(  \rho_{B} - \sigma_B \right) ]$ in terms of the relative entropy of $\rho_{AB}$ and $\sigma_{AB}$ times a constant depending only on $L(\sigma_{AB})$ (since $L(\sigma_{AB})$ represents how entangled $\sigma_{AB}$ is between the regions $A$ and $B$).
The first well-known result we will use in this step is Pinsker's inequality \cite{csiszar, pinsker}.
\begin{thm}\label{thm:pinsker}
For $\rho_{AB}$ and $\sigma_{AB}$ density matrices, it holds that
\begin{equation}
\norm{\rho_{AB}-\sigma_{AB}}_1^2 \le 2 \ent{\rho_{AB}}{\sigma_{AB}}.
\end{equation}
\end{thm}
This result will be of use at the end of the proof to finally obtain the relative entropy in the right-hand side of the desired inequality. However, it is important to notice the different scales of the $\LLL^1$-norm of the difference between $\rho_{AB}$ and $\sigma_{AB}$  and the relative entropy of $\rho_{AB}$ and $\sigma_{AB}$ in Pinsker's inequality. Since we are interested in obtaining the relative entropy with exponent one, we will need to increase the degree of the term with the trace we already have and from which we will construct an $\LLL^1$-norm (since, for the moment, its degree is one). We will see later that the fact that in $ \tr[L (\sigma_{AB}) \left( \rho_{A} - \sigma_A \right) \otimes \left(  \rho_{B} - \sigma_B \right) ]$ we have $ \left( \rho_{A} - \sigma_A \right) \otimes \left(  \rho_{B} - \sigma_B \right) $ split into two regions, the multiplicativity of the trace with respect to tensor products and the monotonicity of the relative entropy play a decisive role 
 in the proof.

Another important fact that we notice in the left-hand side of Pinsker's inequality is that there is a difference between two states (in fact, the ones appearing in the relative entropy). This justifies the use of Lemma \ref{lemma:remark} at the end of Step \ref{step:2}, to obtain something similar to the difference between $\rho_{AB}$ and $\sigma_{AB}$.

We are now ready to prove the third step in the proof of Theorem \ref{thm:quasifactorizationAB}.
\begin{step}\label{step:3}
With the notation of Theorem \ref{thm:quasifactorizationAB},
\begin{equation}
\tr[L (\sigma_{AB}) \left( \rho_{A} - \sigma_A \right) \otimes \left(  \rho_{B} - \sigma_B \right) ]\le 2 \norm{L(\sigma_{AB})}_\infty \ent {\rho_{AB}} {\sigma_{AB}}.
\end{equation}
\end{step}
\begin{proof}

We use the multiplicativity with respect to tensor products of the trace norm and Hölder's inequality between the trace norm and the operator norm. Thus,
\begin{eqnarray*}
\tr[ L(\sigma_{AB}) \, (\rho_{A} - \sigma_{A})\otimes(\rho_{B}-\sigma_{B}) ]&\le & \norm{L(\sigma_{AB})}_\infty \norm{ \left( \rho_{A} - \sigma_{A} \right) \otimes \left( \rho_{B}-\sigma_{B} \right) }_1\\ 
 &=& \norm{L(\sigma_{AB})}_\infty \norm{\rho_{A} - \sigma_{A}}_1 \norm{\rho_{B}-\sigma_{B}}_1.
\end{eqnarray*} 

Finally, Pinsker's inequality (Theorem \ref{thm:pinsker}) implies that
\begin{center}
$\ds \norm{\rho_{A} - \sigma_{A}}_1 \leq \sqrt{ 2 \ent{\rho_{A}}{\sigma_{A}} }, \phantom{sadda} \norm{\rho_{B} - \sigma_{B}}_1 \leq \sqrt{ 2 \ent{\rho_{B}}{\sigma_{B}} } $.
\end{center}

Therefore,
\begin{center}
$\ds  \norm{\rho_{A} - \sigma_{A}}_1 \norm{\rho_{B}-\sigma_{B}}_1\le
 2 \sqrt{\ent{\rho_{A}}{\sigma_{A}} \ent{\rho_{B}}{\sigma_{B}}} \le 2 \ent {\rho_{AB}} {\sigma_{AB}},$
\end{center}
where in the last inequality we have used monotonicity of the relative entropy with respect to the partial trace (Proposition \ref{prop:REprop}). 
\end{proof}

If we now put together Steps \ref{step:1}, \ref{step:2} and \ref{step:3}, we obtain the following expression 
\begin{equation}
(1 + 2 \norm{L(\sigma_{AB})}_\infty) \ent{\rho_{AB}}{\sigma_{AB}}\ge \ent{\rho_A}{\sigma_A} + \ent{\rho_B}{\sigma_B},
\end{equation}
with
\begin{center}
$\ds L(\sigma_{AB})  = \mathcal{T}_{\sigma_{A} \otimes \sigma_{B}} \left(\sigma_{AB} \right)  - \identity_{AB}$.
\end{center}

This inequality already constitutes a quantitative extension of (\ref{superadditivity}) for arbitrary density operators $\sigma_{AB}$ in the sense that if $\sigma_{AB}$ is a tensor product between $A$ and $B$, we recover the usual superadditivity, and in general $\norm{L(\sigma_{AB})}_\infty$ measures how far $\sigma_{AB}$ is from $\sigma_A \otimes \sigma_B$. 

In the fourth and final step of the proof, we bound $\norm{L(\sigma_{AB})}_\infty$  by 
\begin{center}
$\norm{ \sigma_A^{-1/2} \otimes \sigma_B^{-1/2} \, \sigma_{AB} \, \sigma_A^{-1/2} \otimes \sigma_B^{-1/2} - \identity_{AB}}_\infty$,
\end{center}
a quantity from which the closeness to $0$ whenever $\sigma_{AB}$ is near from being a tensor product is directly deduced. It also has some physical interpretation in quantum many body systems that will be discussed after proving Step \ref{step:4}.

 First, we need to introduce the setting of non-commutative $\LLL^p$ spaces with a $\rho$-weighted norm \cite{kosaki}. The central property of these non-commutative ${\LLL}^p$ spaces is that they are equipped with a \textit{weighted norm} which, for a full rank state $\rho\in \SSS_{AB}$, is given by 
\begin{center}
$\ds \norm{f}_{\LLL^p(\rho)} := \tr[\abs{\rho^{1/2p} f \rho^{1/2p}}^p]^{1/p}  \, $ for every $\,f \in \A_{AB}$. 
\end{center} 

Analogously, the $\rho$-\textit{weighted inner product} is given by
\begin{center}
$\ds \left\langle f, g \right\rangle_\rho := \tr[\sqrt{\rho} f \sqrt{\rho} g]  \, $ for every $\,f,g \in \A_{AB}$. 
\end{center}

Some fundamental properties of these spaces are collected in the following proposition.

\begin{prop}\label{prop:noncomLp} 

Let $\rho \in \SSS_{AB}$. The following properties hold for $\rho$-weighted norms:
\begin{enumerate}
\item \textbf{Order.} $\forall p,q \in [1, \infty)$, with $p \leq q$, we have $\norm{f}_{\LLL^p(\rho)} \leq \norm{f}_{\LLL^q(\rho)} \, \forall f \in \A_{AB}$. 

\item \textbf{Duality.} $\forall f \in \A_{AB}$, we have $\norm{f}_{\LLL^p(\rho)} = \sup \qty{ \left\langle g, f \right\rangle_\rho, g \in \A_{AB}, \norm{g}_{\LLL^q(\rho)} \leq 1 }$ for $1/p+1/q=1$.

\item \textbf{Operator norm.} $\forall f \in \A_{AB}$, we have $\norm{f}_{\LLL^\infty(\rho)} = \norm{f}_\infty$, the usual operator norm.
\end{enumerate}
\end{prop}

Another tool we will use in the proof of Step \ref{step:4} is the following result.

\begin{lemma}\label{lemma:contract}
Consider $\rho \in \SSS_{AB}$ and let $T$ be a quantum channel verifying $T^*(\rho)=\rho$, where $T^*$ denotes the dual of $T$ with respect to the Hilbert-Schmidt scalar product. Then, $T$ is contractive between $\LLL^1(\rho)$ and $\LLL^1(\rho)$, i.e., the following inequality holds for every $X \in \BB_{AB}$:
\begin{equation}
\norm{T(X)}_{\LLL^1(\rho)} \leq \norm{X}_{\LLL^1(\rho)}.
\end{equation}
\end{lemma}

\begin{proof}
Using the property of duality for the $\rho$-weighted norms of $\LLL^p$-spaces (property 2 of Proposition \ref{prop:noncomLp}), we can write:
\begin{align*}
\ds \norm{T(X)}_{\LLL^1(\rho)}  &=  \un{\norm{Y}_{\LLL^\infty(\rho)} \leq 1}{\text{sup}} \tr[T(X)\,\rho^{1/2} \, Y \, \rho^{1/2}] \\
&= \un{\norm{Y}_\infty \leq 1}{\text{sup}} \tr[T(X)\,\rho^{1/2} \, Y \, \rho^{1/2}] \\
&= \un{- \identity \leq Y \leq \identity }{\text{sup}} \tr[T(X)\, \rho^{1/2} \, Y \, \rho^{1/2}],
\end{align*}
where in the first step we have used the fact that, for every $\rho \in \SSS_{AB}$, $\norm{\cdot}_{\LLL^\infty(\rho)}$ coincides with the operator norm.

Recalling now that $T^*$ is the dual of $T$ with respect to the Hilbert-Schmidt scalar product, we have:
\begin{align*}
\ds \tr[T(X)\, \rho^{1/2} \, Y \, \rho^{1/2}] &= \tr[X\, T^*( \rho^{1/2} \, Y \, \rho^{1/2})] \\
&= \tr[X\, \rho^{1/2}  \, \rho^{-1/2} \,  T^*( \rho^{1/2} \, Y \, \rho^{1/2}) \,  \rho^{-1/2} \,  \rho^{1/2}  ].
\end{align*}

Since we are considering the supremum over the observables verifying $-\identity \leq Y \leq \identity$, if we apply to these inequalitites $T^*(\rho^{1/2} \cdot \rho^{1/2})$, we have $-\rho \leq T^*(\rho^{1/2} \, Y \, \rho^{1/2}) \leq \rho $ (because of the assumption $T^*(\rho)= \rho$).

Hence, if we denote $Z= \rho^{-1/2} \,  T^*( \rho^{1/2} \, Y \, \rho^{1/2}) \,  \rho^{-1/2}$, it is clear that whenever $ -\identity \leq Y \leq \identity$, also $-\identity \leq Z \leq \identity$. Therefore,
\begin{align*}
\norm{T(X)}_{\LLL^1(\rho)} &= \un{- \identity \leq Y \leq \identity }{\text{sup}} \tr[T(X)\, \rho^{1/2} \, Y \, \rho^{1/2}] \\
&= \un{- \identity \leq Y \leq \identity }{\text{sup}} \tr[X\, \rho^{1/2}  \, \rho^{-1/2} \,  T^*( \rho^{1/2} \, Y \, \rho^{1/2}) \,  \rho^{-1/2} \,  \rho^{1/2}  ] \\
&\leq \un{- \identity \leq Z \leq \identity }{\text{sup}} \tr[X\, \rho^{1/2}  \, Z \,  \rho^{1/2}  ] \\
&= \norm{X}_{\LLL^1(\rho)} ,
\end{align*}
where the last equality comes again from the property of duality of weighted $\LLL^p$-norms.

\end{proof}

In the proof of the previous lemma we have made strong use of the property of duality of $\LLL^p(\rho)$. Indeed, considering the $\LLL^1(\rho)$-norm as dual of the operator norm, has been essential to obtain the desired result.  Using similar tools, we can now prove the last step in the proof of Theorem \ref{thm:quasifactorizationAB}.

\begin{step}\label{step:4}
With the notation of the previous steps, we have
\begin{equation}
\norm{L(\sigma_{AB})}_\infty \leq \norm{ \sigma_A^{-1/2} \otimes \sigma_B^{-1/2} \, \sigma_{AB} \, \sigma_A^{-1/2} \otimes \sigma_B^{-1/2} - \identity_{AB}}_\infty.
\end{equation}
\end{step}

\begin{proof}
The strategy we follow in this proof is the opposite to the one used in the previous lemma, i.e., we study now the $\LLL^\infty(\sigma_A \otimes \sigma_B)$-norm as the dual of  the $\LLL^1(\sigma_A \otimes \sigma_B)$-norm. Since $\norm{\cdot}_{\LLL^\infty (\rho_{AB})}$ coincides with the usual $\infty$-norm (operator norm) for every $\rho_{AB} \in \SSS_{AB}$, we can write
\begin{center}
$\ds  \norm{L(\sigma_{AB})}_\infty = \norm{\mathcal{T}_{\sigma_{A} \otimes \sigma_{B}} \left(\sigma_{AB} \right)  - \identity_{AB}}_{\LLL^\infty(\sigma_A \otimes \sigma_B)}$.
\end{center}

Using the aforementioned property of duality for the $\sigma_A \otimes \sigma_B$-weighted norms of $\LLL^p$-spaces, we have:
\begin{align*}
& \norm{\mathcal{T}_{\sigma_{A} \otimes \sigma_{B}} \left(\sigma_{AB} \right)  - \identity_{AB}}_{\LLL^\infty(\sigma_A \otimes \sigma_B)} = \\
& \phantom{asdasddad}=\un{\norm{O_{AB}}_{\LLL^1(\sigma_A \otimes \sigma_B)} \leq 1}{\text{sup}} \left\langle O_{AB} , \mathcal{T}_{\sigma_{A} \otimes \sigma_{B}} \left(\sigma_{AB} \right)  - \identity_{AB} \right\rangle_{\sigma_A \otimes \sigma_B} \\
&\phantom{asdasddad} =\un{\norm{O_{AB}}_{\LLL^1(\sigma_A \otimes \sigma_B)} \leq 1}{\text{sup}} \tr[ (\sigma_A \otimes \sigma_B)^{1/2} \, O_{AB}\, (\sigma_A \otimes \sigma_B)^{1/2} \left(  \mathcal{T}_{\sigma_{A} \otimes \sigma_{B}} \left(\sigma_{AB} \right)  - \identity_{AB} \right)] \\
&\phantom{asdasddad} =\un{\norm{O_{AB}}_{\LLL^1(\sigma_A \otimes \sigma_B)} \leq 1}{\text{sup}} \left(   \underbrace{\tr[ \sigma_A^{1/2} \otimes \sigma_B^{1/2} \, O_{AB}\, \sigma_A^{1/2} \otimes \sigma_B^{1/2}  \mathcal{T}_{\sigma_{A} \otimes \sigma_{B}} \left(\sigma_{AB} \right) ]}_{R} \right. \\
& \left.  \phantom{asdasddadaszxczzxdasdasdads} - \underbrace{\tr[\sigma_A^{1/2} \otimes \sigma_B^{1/2} \, O_{AB}\, \sigma_A^{1/2} \otimes \sigma_B^{1/2}  ] }_{S}  \right).
\end{align*}

Let us analyze the terms $R$ and $S$ separately. For $R$, we have:
\begin{align*}
R&= \tr[ \sigma_A^{1/2} \otimes \sigma_B^{1/2} \, O_{AB}\, \sigma_A^{1/2} \otimes \sigma_B^{1/2}  \mathcal{T}_{\sigma_{A} \otimes \sigma_{B}} \left(\sigma_{AB} \right) ] \\
& = \tr[ (\sigma_A \otimes \sigma_B)^{1/2} \, O_{AB}\, (\sigma_A \otimes \sigma_B)^{1/2}  \int_{- \infty}^\infty dt \, \beta_0 (t) \left( \sigma_A \otimes \sigma_B \right)^{\frac{-1-it}{2}} \, \sigma_{AB} \,  \left( \sigma_A \otimes \sigma_B \right)^{\frac{-1+it}{2}} ] \\
& = \tr[  O_{AB}\, \int_{- \infty}^\infty dt \, \beta_0 (t) \left( \sigma_A \otimes \sigma_B \right)^{\frac{-it}{2}} \, \sigma_{AB} \,  \left( \sigma_A \otimes \sigma_B \right)^{\frac{it}{2}} ] \\
&  =\int_{- \infty}^\infty dt \, \beta_0 (t) \tr[O_{AB} \, \left( \sigma_A \otimes \sigma_B \right)^{\frac{-it}{2}} \, \sigma_{AB} \,  \left( \sigma_A \otimes \sigma_B \right)^{\frac{it}{2}}] \\
&  =\int_{- \infty}^\infty dt \, \beta_0 (t) \tr[\left( \sigma_A \otimes \sigma_B \right)^{\frac{it}{2}} \; O_{AB} \, \left( \sigma_A \otimes \sigma_B \right)^{\frac{-it}{2}} \, \sigma_{AB} ] \\
&  = \tr[ \sigma_{AB} \underbrace{ \int_{- \infty}^\infty dt \, \beta_0 (t) \left( \sigma_A \otimes \sigma_B \right)^{\frac{it}{2}} \, O_{AB} \,  \left( \sigma_A \otimes \sigma_B \right)^{\frac{-it}{2}}}_{\widetilde{O}_{AB}} ],
\end{align*}
where in the third and last equality we have used the fact that the integral and the trace commute, and the fourth equality is due to the cyclicity of the trace. We have also defined:
\begin{center}
$\ds  \widetilde{O}_{AB}:=  \int_{- \infty}^\infty dt \, \beta_0 (t) \left( \sigma_A \otimes \sigma_B \right)^{\frac{it}{2}} \, O_{AB} \,  \left( \sigma_A \otimes \sigma_B \right)^{\frac{-it}{2}}$.
\end{center}

If we were able to express  $S$ in terms of $\widetilde{O}_{AB}$, we could simplify the expression that appears in the supremum above. We can do that in the following way:

\begin{align*}
S& = \tr[\sigma_A^{1/2} \otimes \sigma_B^{1/2} \, O_{AB}\, \sigma_A^{1/2} \otimes \sigma_B^{1/2}  ] \\
&= \tr[\sigma_A^{1/2} \otimes \sigma_B^{1/2} \, O_{AB}\, \sigma_A^{1/2} \otimes \sigma_B^{1/2}  \int_{- \infty}^\infty dt \, \beta_0 (t) ] \\
&=  \int_{- \infty}^\infty dt \, \beta_0 (t) \tr[\sigma_A^{1/2} \otimes \sigma_B^{1/2} \, O_{AB}\, \sigma_A^{1/2} \otimes \sigma_B^{1/2}  ] \\
&=  \int_{- \infty}^\infty dt \, \beta_0 (t) \tr[(\sigma_A\otimes \sigma_B) \, (\sigma_A\otimes \sigma_B)^{\frac{it}{2}} \, O_{AB}\, (\sigma_A\otimes \sigma_B)^{\frac{-it}{2}}   ] \\
&= \tr[(\sigma_A\otimes \sigma_B) \int_{- \infty}^\infty dt \, \beta_0 (t)  (\sigma_A\otimes \sigma_B)^{\frac{it}{2}} \, O_{AB}\, (\sigma_A\otimes \sigma_B)^{\frac{-it}{2}}   ] \\
&= \tr[(\sigma_A\otimes \sigma_B) \, \widetilde{O}_{AB} ],
\end{align*}
where we have used again the properties of cyclicity of the trace and commutativity of the integral and the trace.

Placing now the values for $R$ and $S$ that we have just computed in the supremum of the first part of the proof, we have:
\begin{align*}
\norm{\mathcal{T}_{\sigma_{A} \otimes \sigma_{B}} \left(\sigma_{AB} \right)  - \identity_{AB}}_{\LLL^\infty(\sigma_A \otimes \sigma_B)} &= \un{\norm{O_{AB}}_{L^1(\sigma_A \otimes \sigma_B)} \leq 1}{\text{sup}} \left(  \tr[\sigma_{AB}  \, \widetilde{O}_{AB} ] - \tr[\sigma_A \otimes \sigma_B \, \widetilde{O}_{AB} ] \right) \\
&=  \un{\norm{O_{AB}}_{L^1(\sigma_A \otimes \sigma_B)} \leq 1}{\text{sup}}   \tr[\, \widetilde{O}_{AB} \left(  \sigma_{AB} - \sigma_A \otimes \sigma_B \right) ] .
\end{align*}

This expression looks much simpler than the one we had before. However, we need to prove that $\norm{\widetilde{O}_{AB}}_{L^1(\sigma_A \otimes \sigma_B)} \leq 1$ in order to see $\widetilde{O}_{AB}$ as one of the terms where the supremum is taken. Indeed, if we consider the map $T: \A_{AB} \rightarrow \A_{AB}$ given by 
\begin{center}
$\ds  O_{AB} \mapsto \int_{- \infty}^\infty dt \, \beta_0 (t)  (\sigma_A\otimes \sigma_B)^{\frac{it}{2}} \, O_{AB}\, (\sigma_A\otimes \sigma_B)^{\frac{-it}{2}}  $,
\end{center}
it is clearly a quantum channel and also verifies $T^*(\sigma_A \otimes \sigma_B)= \sigma_A \otimes \sigma_B$. Hence, in virtue of Lemma \ref{lemma:contract}, we have
\begin{center}
$\ds  \norm{\widetilde{O}_{AB}}_{\LLL^1(\sigma_A \otimes \sigma_B)} \leq \norm{O_{AB}}_{\LLL^1(\sigma_A \otimes \sigma_B)}  $,
\end{center}
and, therefore,
\begin{center}
$\ds  \un{\norm{O_{AB}}_{L^1(\sigma_A \otimes \sigma_B)} \leq 1}{\text{sup}}   \tr[\widetilde{O}_{AB} \left(  \sigma_{AB} - \sigma_A \otimes \sigma_B \right) ] \leq  \un{\norm{\Omega_{AB}}_{L^1(\sigma_A \otimes \sigma_B)} \leq 1}{\text{sup}}   \tr[\Omega_{AB} \left(  \sigma_{AB} - \sigma_A \otimes \sigma_B \right) ]. $
\end{center}

In this last supremum over elements of $1$-norm, we can undo the previous transformations in order to obtain again an $\infty$-norm. First, we need to write the term in the supremum as a $\sigma_A \otimes \sigma_B$-product of two terms:
\begin{align*}
& \tr[\Omega_{AB} \left(  \sigma_{AB} - \sigma_A \otimes \sigma_B \right) ] = \\
&  \phantom{asdasdd}  = \tr[  (\sigma_A \otimes \sigma_B)^{1/2} \, (\sigma_A \otimes \sigma_B)^{-1/2} \, \sigma_{AB}  \, (\sigma_A \otimes \sigma_B)^{-1/2} \, (\sigma_A \otimes \sigma_B)^{1/2} \, \Omega_{AB} ] \\
& \phantom{asdasddas} - \tr[  (\sigma_A \otimes \sigma_B)^{1/2} \,  \Omega_{AB} \, (\sigma_A \otimes \sigma_B)^{1/2} ] \\
& \phantom{asdasdd} =   \left\langle \Omega_{AB} , (\sigma_A \otimes \sigma_B)^{-1/2} \, \sigma_{AB}  \, (\sigma_A \otimes \sigma_B)^{-1/2} \right\rangle_{\sigma_A \otimes \sigma_B}  -\left\langle \Omega_{AB} , \identity_{AB}  \right\rangle_{\sigma_A \otimes \sigma_B}\\
& \phantom{asdasdd} =  \left\langle \Omega_{AB} ,  (\sigma_A \otimes \sigma_B)^{-1/2} \, \sigma_{AB}  \, (\sigma_A \otimes \sigma_B)^{-1/2} -\identity_{AB} \right\rangle_{\sigma_A \otimes \sigma_B}. 
\end{align*}

Finally, using again the property of duality for the norms of  $\LLL^1(\sigma_A \otimes \sigma_B)$ and $ \LLL^\infty(\sigma_A \otimes \sigma_B)$,  we have:
\begin{align*}
&  \un{\norm{\Omega_{AB}}_{L^1(\sigma_A \otimes \sigma_B)} \leq 1}{\text{sup}}   \tr[\Omega_{AB} \left(  \sigma_{AB} - \sigma_A \otimes \sigma_B \right) ] \\
& \phantom{asdasdad} =  \un{\norm{\Omega_{AB}}_{L^1(\sigma_A \otimes \sigma_B)} \leq 1}{\text{sup}}   \left\langle \Omega_{AB} , (\sigma_A \otimes \sigma_B)^{-1/2} \, \sigma_{AB}  \, (\sigma_A \otimes \sigma_B)^{-1/2} - \identity_{AB} \right\rangle_{\sigma_A \otimes \sigma_B} \\
& \phantom{asdasdad} = \norm{\sigma_A^{-1/2} \otimes \sigma_B^{-1/2} \, \sigma_{AB} \, \sigma_A^{-1/2} \otimes \sigma_B^{-1/2} - \identity_{AB}}_{\LLL^{\infty}(\sigma_A \otimes \sigma_B)}  \\
& \phantom{asdasdad}= \norm{\sigma_A^{-1/2} \otimes \sigma_B^{-1/2} \, \sigma_{AB} \, \sigma_A^{-1/2} \otimes \sigma_B^{-1/2} - \identity_{AB}}_{\infty},
\end{align*}
where we have used again the fact  that $\norm{\cdot}_{\LLL^\infty (\rho_{AB})}$ coincides with the usual $\infty$-norm  for every $\rho_{AB} \in \SSS_{AB}$.

In conclusion,
\begin{center}
$\ds \norm{\mathcal{T}_{\sigma_{A} \otimes \sigma_{B}} \left(\sigma_{AB} \right)  - \identity_{AB}}_{\infty} \leq \norm{\sigma_A^{-1/2} \otimes \sigma_B^{-1/2} \, \sigma_{AB} \, \sigma_A^{-1/2} \otimes \sigma_B^{-1/2} - \identity_{AB}}_{\infty}. $
\end{center}

\end{proof}

By putting together Step \ref{step:1}, Step \ref{step:2}, Step \ref{step:3} and Step \ref{step:4}, we
conclude the proof of Theorem \ref{thm:quasifactorizationAB}.

\begin{remark}
This result constitutes an extension of the superadditivity property, i.e., the constant $H(\sigma_{AB})$ that appears in the statement of the main theorem is $0$ when $\sigma_{AB}=\sigma_A \otimes \sigma_B$ and is small whenever $\sigma_{AB} \sim \sigma_A \otimes \sigma_B$. A trivial upper bound can be found with respect to the trace distance as follows,
\begin{align*}
& \norm{\sigma_A^{-1/2} \otimes \sigma_B^{-1/2} \, \sigma_{AB} \, \sigma_A^{-1/2} \otimes \sigma_B^{-1/2} - \identity_{AB}}_{\infty} = \\
& \phantom{asdasdaad} = \norm{\sigma_A^{-1/2} \otimes \sigma_B^{-1/2} ( \sigma_{AB} - \sigma_A \otimes \sigma_B) \sigma_A^{-1/2} \otimes \sigma_B^{-1/2}}_\infty \\
&  \phantom{asdasdaad}  \leq \norm{\sigma_A^{-1/2} \otimes \sigma_B^{-1/2} ( \sigma_{AB} - \sigma_A \otimes \sigma_B) \sigma_A^{-1/2} \otimes \sigma_B^{-1/2}}_1 \\
& \phantom{asdasdaad}  \leq \norm{\sigma_A^{-1/2} \otimes \sigma_B^{-1/2} }_\infty \norm{\sigma_{AB} - \sigma_A \otimes \sigma_B}_1 \norm{\sigma_A^{-1/2} \otimes \sigma_B^{-1/2} }_\infty \\
& \phantom{asdasdaad}  \leq \sigma_{\text{min}}^{-2} \, \norm{\sigma_{AB} - \sigma_A \otimes \sigma_B}_1 .
\end{align*}

\end{remark}

\begin{remark}
The term  $\norm{H(\sigma_{AB})}_\infty$ is also closely related to certain forms of \textit{decay of correlations} of states that have already appeared in quantum many body systems, such as \textit{LTQO (Local Topological Quantum Order)} \cite{spiros}, or the concept of \textit{local indistinguishability} as a strengthened form of \textit{weak clustering} in \cite{kast-brand}. 

Let us suppose that $\norm{H(\sigma_{AB})}_\infty \leq \lambda(\ell)$ for a certain small scalar $\lambda(\ell)$ that decays sufficiently fast as a function of the distance $\ell$ between regions $A$ and $B$ in a many body system, and denote by $\left\langle f \right\rangle_\varphi$ the expected value of an observable $f \in \A_{AB} $ with respect to a state $\varphi$ (usually the ground or thermal state of the system). Then, for every observable of the form $O_A \otimes O_B \geq 0$, if the reduced density matrix on $AB$ of $\varphi$ is $\sigma_{AB}$, the previous condition can be rewritten as
\begin{center}
$\ds  \abs{\left\langle O_A O_B \right\rangle_\varphi - \left\langle O_A \right\rangle_\varphi \left\langle O_B \right\rangle_\varphi } \leq \lambda \left\langle O_A  \right\rangle_\varphi \left\langle O_B \right\rangle_\varphi $.
\end{center}

One can now compare this expression with the definition of decay of correlations 
\begin{center}
$\ds \abs{\left\langle O_A O_B \right\rangle_\varphi - \left\langle O_A \right\rangle_\varphi \left\langle O_B \right\rangle_\varphi } \leq \lambda(\ell) \norm{O_A}_\infty \norm{O_B}_\infty$,
\end{center}
or LTQO  
\begin{center}
$\ds \abs{\left\langle O_A O_B \right\rangle_\varphi - \left\langle O_A \right\rangle_\varphi \left\langle O_B \right\rangle_\varphi } \leq \lambda(\ell) \left\langle O_A \right\rangle_\varphi \norm{O_B}_\infty$.
\end{center}

\end{remark}

\section{Conclusion}
In this work, we have proven an extension of the property of superadditivity of the quantum relative entropy for general states. Our result constitutes an improvement to the usual lower bound for the relative entropy of two bipartite states, given by the property of monotonicity, in terms of the relative entropies in the two constituent spaces, whenever the second state is near to be a tensor product. Therefore, it might be relevant for situations where this property is expected to hold, such as quantum many body systems, in which  it is likely that the Gibbs state satisfies this property in spatially separated systems. 

In \cite{kast-brand}, Kastoryano and Brandao proved, for certain Gibbs samplers, the existence of a positive spectral gap for the dissipative dynamics, via a quasi-factorization result of the variance. This provides a bound for the mixing time of the evolution of the semigroup that drives the system to thermalization which is polynomial in the system size. We leave for future work the possibility of using the result of the present paper to obtain a quasi-factorization of the relative entropy in quantum many body systems, which could allow us to prove, under some conditions of decay of correlations on the Gibbs state, the existence of a positive log-Sobolev constant,  obtaining an exponential improvement in the bound for the mixing time obtained in \cite{kast-brand}.

\section*{Acknowledgment}

We are very grateful to D. Sutter and M. Tomamichel, who detected an error in a previous version of the paper. We also thank M. Junge for fruitful discussions. AC and DPG acknowledge support from MINECO (grant MTM2014-54240-P), from Comunidad de Madrid (grant QUITEMAD+- CM, ref. S2013/ICE-2801), and the European Research Council (ERC) under the European Union’s Horizon 2020 research and innovation programme (grant agreement No 648913). AC is partially supported by a La Caixa-Severo Ochoa grant (ICMAT Severo Ochoa project SEV-2011-0087, MINECO). AL acknowledges financial support from the European Research Council (ERC Grant Agreement no 337603), the Danish Council for Independent Research (Sapere Aude) and VILLUM FONDEN via the QMATH Centre of Excellence (Grant No. 10059). This work has been partially supported by ICMAT Severo Ochoa project SEV-2015-0554 (MINECO).

\appendix


\end{document}